\documentclass[pra,twocolumn,showpacs,amsmath,amssymb,superscriptaddress,floatfix,nofootinbib]{revtex4-1}
\usepackage{graphicx,amsmath,amssymb,amsfonts,dsfont,natbib,physics}
\usepackage{multirow}
\usepackage{lipsum,xcolor}
\usepackage{verbatim}
\usepackage{commath}
\usepackage{bm}
\usepackage{float}
\usepackage{pifont}

\usepackage{psfrag}

\usepackage{bm}
\usepackage{graphicx}
\usepackage{subfigure}
\usepackage{color}
\usepackage{amsthm}

\newcommand{\beq}{\begin{equation}}
\newcommand{\eneq}{\end{equation}}

\newcommand{\twopartdef}[4]
{
	\left\{
		\begin{array}{ll}
			#1 & \mbox{if } #2 \\
			#3 & \mbox{if } #4
		\end{array}
	\right.
}

\newcommand{\bk}{\boldsymbol{k}}
\newcommand{\br}{\boldsymbol{r}}
\newcommand{\bp}{\boldsymbol{\pi}}

\newcommand{\mE}{\mathcal{E}}

\newcommand{\mS}{\mathcal{S}}
\newcommand{\hgen}{H_{\textrm{gen}}}
\newcommand{\hhub}{H_{\textrm{Hub}}}
\newcommand{\hh}{\widehat{h}}

\newcommand{\hxy}{H^{(x)}}
\newcommand{\haklt}{H^{(a)}}
\newcommand{\pxy}{\mathcal{P}^{(x)}}
\newcommand{\paklt}{\mathcal{P}^{(a)}}
\newcommand{\exy}{\mathcal{E}^{(x)}}

\newcommand{\vxy}{\hV^{(x)}}
\newcommand{\vaklt}{\hV^{(a)}}

\newcommand{\prodal}[2]{\underset{#1}{\overset{#2}{\prod}}}
\newcommand{\sumal}[2]{\underset{#1}{\overset{#2}{\sum}}}

\newcommand{\dlangle}{\langle\langle}
\newcommand{\drangle}{\rangle\rangle}

\input{epsf}

\newcommand{\nn}{\nonumber}

\newcommand{\cd}{c^\dagger}
\newcommand{\ed}{\eta^\dagger}

\newcommand{\hT}{\widehat{T}}

\newcommand{\hV}{\widehat{V}}
\newcommand{\hI}{\widehat{I}}

\newcommand{\hn}{\widehat{n}}

\newtheorem{theorem}{Theorem}[section]

\newtheorem{lemma}[theorem]{Lemma}

\begin{document}
\tolerance 10000
\title{Eta-Pairing in Hubbard Models: From Spectrum Generating Algebras to Quantum Many-Body Scars}
\author{Sanjay Moudgalya}
\affiliation{Department of Physics, Princeton University, NJ 08544, USA}

\author{Nicolas Regnault}
\affiliation{Department of Physics, Princeton University, NJ 08544, USA}
\affiliation{Laboratoire de Physique de l'Ecole normale sup\'{e}rieure, ENS, Universit\'{e} PSL, CNRS, Sorbonne Universit\'{e}, Universit\'{e} Paris-Diderot, Sorbonne Paris Cit\'{e}, Paris, France}

\author{B. Andrei Bernevig}
\affiliation{Department of Physics, Princeton University, NJ 08544, USA}

\begin{abstract}
We revisit the $\eta$-pairing states in Hubbard models and explore their connections to quantum many-body scars to discover a universal scars mechanism.
$\eta$-pairing occurs due to an algebraic structure known as a Spectrum Generating Algebra (SGA), giving rise to equally spaced towers of eigenstates in the spectrum. 
We generalize the original $\eta$-pairing construction and show that several Hubbard-like models on arbitrary graphs exhibit SGAs, including ones with disorder and spin-orbit coupling. 
We further define a Restricted Spectrum Generating Algebra (RSGA) and give examples of perturbations to the Hubbard-like models that preserve an equally spaced tower of the original model as eigenstates. 
The states of the surviving tower exhibit a sub-thermal entanglement entropy, and we analytically obtain parameter regimes for which they lie in the bulk of the spectrum, showing that they are exact quantum many-body scars.
The RSGA framework also explains the equally spaced towers of eigenstates in several well-known models of quantum scars, including the AKLT model. 
\end{abstract}
\maketitle
\date{\today}


\section{Introduction}\label{sec:intro}
The study of ergodicity and its breaking in isolated quantum systems has been a growing branch of research in quantum many-body physics.
In particular, novel mechanisms for the violation of the Eigenstate Thermalization Hypothesis (ETH)~\cite{deutsch1991quantum, srednicki1994chaos, rigol2008thermalization, d2016quantum} have gained recent attention, particularly a mechanism known as Quantum Many-Body Scarring. 
Unlike other mechanisms of ETH violation such as integrability and Many-Body Localization (MBL)~\cite{rahul2015review}, where the entire spectrum of a Hamiltonian violates ETH, quantum scarred models consist of some ETH-violating eigenstates in an otherwise ETH-satisfying spectrum. 
Hamiltonian systems currently known to exhibit quantum scars can be roughly classified into three categories. 
First, models that exhibit an analytically solvable equally spaced tower of eigenstates.
This line of study was initiated by the discovery of an ETH-violating tower of states in the celebrated AKLT model~\cite{aklt1987,Moudgalya2018a, Moudgalya2018b}.
Similar towers of eigenstates were subsequently discovered in several families of models~\cite{schecter2019weak, iadecola2019quantum2, chattopadhyay2019quantum, shibata2019onsagers, mark2020unified, moudgalya2020large}.
Second, ETH-violating eigenstates can be systematically embedded within an otherwise non-integrable model, as first illustrated in Ref.~\cite{mori2017eth}.
This formalism can be used to explain the existence of scars in several models~\cite{ok2019topological, schecter2019weak,lee2020exact, surace2020weakergodicitybreaking}. 
Third, approximate quantum scars that manifest in the dynamics of simple initial states, first found in the PXP model~\cite{turner2017quantum, turner2018quantum, khemani2019int, ho2018periodic, iadecola2019quantum, surace2019lattice}.
Similar phenomenology has also been found in a variety of models in one~\cite{choi2018emergent, bull2020quantum, michailidis2019slow,2bull2019scar, moudgalya2019quantum, hudomal2019, alhambra2019revivals, robinson2019signatures, james2019nonthermal, sinha2019chaos} and higher~\cite{michailidis2020stabilizing, lin2020quantum, voorden2020quantum} dimensions.
Unlike the other two categories of quantum scars, the towers of states in these models are not exactly solvable, although some different eigenstates can be analytically obtained in some cases~\cite{lin2018exact,shiraishi2019connection,mark2019new,lin2020quantum, voorden2020quantum}.  
However, this latter situation is the only case where the effect of quantum scars has been experimentally observed - in a cold atom experiment~\cite{bernien2017probing}.
In this work, we focus on the analytically tractable quantum scars of the first category, i.e. equally spaced towers of states.
All the known examples of such towers exist in hard-core bosonic spin models, which are hard to naturally realize in experiments. 
It is thus highly desirable to look for similar phenomena in more physically relevant electronic systems. 
Equally spaced towers of states have been known to occur in the celebrated Hubbard models since the seminal work of Yang that introduced the mechanism of $\eta$-pairing~\cite{yang1989eta}.
The existence of $\eta$-pairing and the related Off-Diagonal Long-Range Order (ODLRO) is attributed to the understanding of a pseudospin $SU(2)$ symmetry of the Hubbard model~\cite{zhang1990pseudospin, yang1990so}. 
There has since been a vast amount of literature studying the existence and properties of $\eta$-pairing and its generalizations to a wide range of models~\cite{yang1992remarks, shen1993exact, essler1995extended, deboer1995eta, albertini1995xxz, schadschneider1995superconductivity,fan1999the, fan2005entanglement, zhai2005two, li2019exact}.
Despite this large body of literature, the natural connection between the $\eta$-pairing states and the infinite-temperature quantum dynamics of the Hubbard models has not been extensively explored apart from the one-dimensional case, where the Hubbard model is fully integrable~\cite{essler2005one}.
Notable exceptions include Ref.~\cite{vafek2017entanglement}, which computed the entanglement of some analytically tractable eigenstates~\cite{yang1990so} of the $D$-dimensional Hubbard models, and Ref.~\cite{yu2018beyond}, where the effect of $\eta$-pairing on many-body localized Hubbard models was numerically explored. 
However, the analytically tractable $\eta$-pairing states in the $D$-dimensional Hubbard models~\cite{yang1990so} are not examples of quantum scars even though some of them have low entanglement since it was proven that they are the only eigenstates in their respective quantum number sectors~\cite{vafek2017entanglement}. 
That is, they do not appear to be mixed with ETH-satisfying states in the spectrum with the same set of quantum numbers. 
Given the similarity with quantum scars, it is natural to explore the precise connection of these $\eta$-pairing towers of states and quantum many-body scars.  
In particular, we ask if it is possible to deform the Hubbard model such that the pseudospin - $\eta$- symmetry is broken (and hence most of the $\eta$-pairing eigenstates would cease to be eigenstates) while preserving a subset of the analytically tractable eigenstates of the Hubbard model, which would then become examples of quantum many-body scars. 
To do so, we first recast $\eta$-pairing as a real-space phenomenon in contrast to the momentum-space approach employed in most of the literature.
This makes clear the minimal conditions necessary for the existence of $\eta$-pairing, and unravels a large class of Hubbard models with disorder and/or spin-orbit coupling that exhibit $\eta$-pairing.
We refer to this algebraic structure as a Spectrum Generating Algebra (SGA).
We then introduce the concept of a Restricted Spectrum Generating Algebra (RSGA), and we show that perturbations can be added to the Hubbard models that enable some of the analytically tractable $\eta$-pairing towers of the Hubbard models to survive as eigenstates of the perturbed models. 
We show analytically that these states, which have a low entanglement entropy \cite{vafek2017entanglement}, lie in the bulk of the spectrum of their quantum number sectors, and thus form examples of quantum many-body scars. 
We show that these RSGAs also appear in existing models of quantum scars in the literature, for example, the AKLT model~\cite{Moudgalya2018a} and the spin-1 XY model~\cite{schecter2019weak}. 
We note that related algebraic structures have appeared in the literature in the past in the context of Generalized Hubbard Models~\cite{deboer1995eta}, and more recently in the context of unifying formalisms for quantum scarred models~\cite{mark2020unified, bull2020quantum}.  
This paper is organized as follows.
In Sec.~\ref{sec:etareview} we review the Fermi-Hubbard model and the existence of a Spectrum Generating Algebra (SGA), i.e. the $\eta$-pairing states.  
In Sec.~\ref{sec:etareal}, we illustrate the generalization of the $\eta$-pairing states to Hubbard models on arbitrary graphs with disorder in the hopping terms and with spin-orbit coupling.
We discuss some examples in Sec.~\ref{sec:examples}.
In Sec.~\ref{sec:towers}, we introduce the concept of Restricted Spectrum Generating Algebra (RSGA), which captures the behavior of several known quantum scarred models, and we introduce perturbations to the (generalized) Hubbard models that realize an RSGA.  
There, we analytically show that the tower of eigenstates realized by the RSGA are quantum many-body scars of the perturbed Hamiltonians by deriving the conditions for which the states are in the bulk of the spectra of their quantum number sectors.
In Sec.~\ref{sec:RSGAquantumscars}, we comment on connections between the RSGA formalism and quantum scarred models in the literature. 
We conclude with a discussion of future directions in Sec.~\ref{sec:conclusions}.
\section{Review of $\boldsymbol{\eta}$-pairing in the Hubbard Model}\label{sec:etareview}
We review the construction of $\eta$-pairing states in the Fermi-Hubbard model (that we also refer to as the ``Hubbard model"), first obtained in Refs.~\cite{yang1989eta, zhang1990pseudospin, yang1990so}.
The Hubbard Hamiltonian is given by
\begin{eqnarray}
    &\hhub = \sumal{\sigma \in \{\uparrow, \downarrow\}}{}{\left[-t\sumal{\langle \br, \br'\rangle}{}{\left(\cd_{\br, \sigma} c_{\br', \sigma} + h.c\right)} -\mu\sumal{\br}{}{\cd_{\br, \sigma} c_{\br, \sigma}}\right]} \nn \\
    &+ U \sumal{\br}{}{\hn_{\br, \uparrow} \hn_{\br, \downarrow}}.
\label{eq:FermiHubbard}
\end{eqnarray}
where $\hn_{\br,\sigma} \equiv \cd_{\br,\sigma} c_{\br,\sigma}$, $\{\br\}$ is the set of sites on an arbitrary graph, (in $D$ dimensions) and $\langle \br, \br' \rangle$ denotes nearest neighboring sites. 
On a $D$-dimensions hypercubic lattice with periodic boundary conditions and even lengths in all directions, the Hubbard model admits translation invariance, charge and spin $SU(2)$ symmetries, and lattice mirror symmetries. 
In that case, the Hubbard Hamiltonian can be written as
\begin{equation}
    \hhub = \sumal{\bk}{}{\sumal{\sigma \in \{\uparrow, \downarrow\}}{}{\mE_{\bk} \cd_{\bk, \sigma} c_{\bk, \sigma}}} + U \sumal{r}{}{\hn_{\br, \uparrow} \hn_{\br, \downarrow}},
\label{eq:HubbardPBC}
\end{equation}
where
\begin{equation}
    \mE_{\bk} \equiv -\mu -2 t\sumal{i = 1}{D}{\cos k_i},
\label{eq:dispersion}
\end{equation} 
where $k_i$ is the momentum in the $i$-th direction. 
For these Hamiltonians, Refs.~\cite{yang1989eta, zhang1990pseudospin} showed that there exists an operator $\ed$ defined as 
\begin{eqnarray}
    \ed \equiv \sumal{\bk}{}{\cd_{\bk,\uparrow} \cd_{\bp-\bk, \downarrow}} = \sumal{\br}{}{e^{i\bp\cdot\br} \cd_{\br, \uparrow} \cd_{\br,\downarrow}},
\label{eq:etaops}
\end{eqnarray}
where $\bp \equiv (\pi, \pi, \cdots, \pi)$, that satisfies the relation
\begin{equation}
    [\hhub, \ed] = (U - 2\mu)\ed. 
\label{eq:HubbardSGA}
\end{equation}
In fact, for a system with $L$ sites in each dimension, the $\ed$ and $\eta$ operators, along with 
\begin{equation}
    \eta_z \equiv \frac{1}{2}[\ed, \eta] = \frac{1}{2}\left(\sumal{\br, \sigma}{}{\hn_{\br, \sigma}} - L^D\right),
\end{equation}
constitute a full (pseudospin) $SU(2)$ symmetry of the Hubbard model~\cite{zhang1990pseudospin}.
That is,
\begin{eqnarray}
    &[\eta_z, \ed] = \ed,\;\;[\eta_z, \eta] = -\eta,\nn \\
    &[\hhub, \eta_z] = 0,\;\;[\hhub, \boldsymbol{\eta}^2] = 0,
\end{eqnarray}
where $\boldsymbol{\eta}^2$ is the total pseudospin operator
\begin{equation}
    \boldsymbol{\eta}^2 \equiv \frac{1}{2}(\ed\eta + \eta \ed) + (\eta_z)^2.
\end{equation}
Eq.~(\ref{eq:HubbardSGA}) is said to be an example of a \textit{Spectrum Generating Algebra} (SGA)~\cite{barut1965dynamical, dothan1965series}, when an operator satisfying\footnote{We note that the terms ``Spectrum Generating Algebra" and ``Dynamical Symmetry"  have been used to denote a variety of related (but subtly distinct) concepts in the literature~\cite{arno1988dynamical, solomon1998coherent, gruber2012symmetries, leviatan2011partial}. 
In this work we will only use SGA to refer to Eq.~(\ref{eq:SGA}).}
\begin{equation}
    [H, \ed] = \mE \ed,
\label{eq:SGA}
\end{equation}
generates a series of equally spaced energy eigenstates.
Indeed, if $\ket{\psi_0}$ is an eigenstate of $H$ with energy $E_0$, $\ed \ket{\psi_0}$ is also an eigenstate with energy $E_0 + \mE$ (see App.~\ref{app:RSGAtower}).
Iterating this idea until $(\ed)^{N+1}\ket{\psi_0}$ vanishes (which it does, as $\ed$ increases the number of fermions by 2), we obtain an equally-spaced tower of states given by
\begin{equation}
    \{\ket{\psi_0}, \ed \ket{\psi_0}, \cdots, (\ed)^N \ket{\psi_0}\} 
\label{eq:psi0tower}
\end{equation}
with corresponding energies given by
\begin{equation}
    \{E_0, E_0 + \mE, E_0 + 2 \mE, \cdots, E_0 + N\mE\}.
\end{equation}
In the case of the Hubbard model, $\mE = (U - 2\mu)$ (see Eq.~(\ref{eq:HubbardSGA})).

Although Eq.~(\ref{eq:SGA}) leads to the existence of a tower of states starting from $\ket{\psi_0}$, it does not imply that the expressions of any of the eigenstates can be obtained analytically. 
However, for the Hubbard Hamiltonian of Eq.~(\ref{eq:HubbardPBC}) in any dimensions, several $U$-independent eigenstates can be obtained analytically~\cite{zhang1990pseudospin}.
Note that the vacuum state $\ket{\Omega}$, and the spin-polarized eigenstates of the hopping operator are also ferromagnetic eigenstates of the Hubbard Hamiltonian.
As a consequence of the spin $SU(2)$ symmetry, multiplets eigenstates can be obtained by applying spin raising and lowering operators on these eigenstates.
Further, several more eigenstates are obtained applying the $\ed$ operator repeatedly on those eigenstates, although not all of the resulting eigenstates are independent~\cite{zhang1990pseudospin}.
\section{$\boldsymbol{\eta}$-pairing on arbitrary graphs}\label{sec:etareal}
To unravel the most general necessary conditions for having a spectrum generating algebra, we break several symmetries of the Hubbard Hamiltonian $\hhub$.  We consider much more general Hubbard Hamiltonians with disorder and spin-orbit coupling on arbitrary graphs.
This forces us to obtain a real-space understanding of $\eta$-pairing, unlike the momentum space derivations used in most of the literature.
We consider the generalized Hubbard Hamiltonian
\begin{eqnarray}
    &\hgen = -\sumal{\sigma, \sigma'}{}{\sumal{\langle\br, \br'\rangle}{}{\underbrace{\left(t^{\sigma,\sigma'}_{ \br, \br'}\cd_{\br, \sigma} c_{\br', \sigma'} + t^{\sigma',\sigma}_{\br', \br}\cd_{\br', \sigma'} c_{\br, \sigma}\right)}_{\equiv \hT^{\sigma,\sigma'}_{\br, \br'}}}} \nn \\
    &- \sumal{\br, \sigma}{}{\mu_{\br, \sigma}\hn_{\br, \sigma}} + \sumal{\br}{}{U_{\br} \hn_{\br, \uparrow} \hn_{\br, \downarrow}},
\label{eq:Hamilfinalmain}
\end{eqnarray}
where $\sigma, \sigma'$ denotes the spin, $\langle \br, \br'\rangle$ denote nearest neighboring sites $\br$ and $\br'$ on the graph, $\{t^{\sigma,\sigma'}_{\br,\br'}\}$ are spin and position dependent hopping strengths satisfying Hermiticity ($t^{\sigma,\sigma'}_{\br,\br'} = t^{\sigma',\sigma\star}_{\br',\br}$), and $\{\mu_{\br, \sigma}\}$ are spin dependent real chemical potentials.
Note that since we are considering the Hamiltonian Eq.~(\ref{eq:Hamilfinalmain}) on arbitrary graphs, we can without loss of generality consider the hopping terms $t^{\sigma,\sigma'}_{\br, \br'}$ to be non-vanishing on nearest neighboring sites on the graph.
%
%
%
%
Note that the Hamiltonian of Eq.~(\ref{eq:Hamilfinalmain}) breaks all the usual symmetries of the original Hubbard model of Eq.~(\ref{eq:FermiHubbard}) except the charge $U(1)$ symmetry.
Despite breaking these symmetries, we find that $\hgen$ admits an SGA (and hence preserves the pseudospin ``$\eta$" symmetry) provided the hopping strengths $\{t^{\sigma,\sigma'}_{\br, \br'}\}$, and $\{\mu_{\br, \sigma}\}$ are appropriately chosen.
We define the $\ed$ operator to be 
\begin{equation}
    \ed \equiv \sumal{\br}{}{q_{\br} \ed_{\br}} \equiv \sumal{\br}{}{q_{\br} \cd_{\br,\uparrow}\cd_{\br,\downarrow}},
\label{eq:etadefn}
\end{equation}
and derive conditions on $\{q_{\br}\}$ and $\{t^{\sigma, \sigma'}_{\br, \br'}\}$ such that the Hamiltonian of Eq.~(\ref{eq:Hamilfinalmain}) admits an SGA. 
We first explicitly compute the following commutators of $\ed$ with the on-site terms of the Hamiltonian 
\begin{eqnarray}
    &[\sumal{\br, \sigma}{}{\mu_{\br, \sigma} \hn_{\br, \sigma}}, \ed] = \sumal{\br}{}{q_{\br}\mu_{\br, \sigma} [\hn_{\br, \sigma}, \ed_{\br}]} = \sumal{\br}{}{q_{\br} \left(\sumal{\sigma}{}{\mu_{\br, \sigma}}\right) \ed_{\br}} \nn \\
    &[\sumal{\br, \sigma}{}{U_{\br} \hn_{\br, \uparrow} \hn_{\br, \downarrow}}, \ed] = \sumal{\br}{}{q_{\br} U_{\br} [\hn_{\br, \uparrow} \hn_{\br, \downarrow}, \ed_{\br}]} = \sumal{\br}{}{q_{\br} U_{\br}\ed_{\br}} \nn \\
\label{eq:etagraphVcomm}
\end{eqnarray}
where we have used Eq.~(\ref{eq:usefuletacoms1}). 
Thus, by choosing on-site chemical potentials and interactions that satisfy 
\begin{equation}
    U_{\br} - \mu_{\br, \uparrow} - \mu_{\br, \downarrow} = \mE,
\label{eq:onsitechemical}
\end{equation}
we obtain
\begin{equation}
    [\sumal{\br, \sigma}{}{\mu_{\br, \sigma} \hn_{\br, \sigma}} + \sumal{\br, \sigma}{}{U_{\br} \hn_{\br, \uparrow} \hn_{\br, \downarrow}}, \ed]  = \mE \ed.
\label{eq:etagraphMVcomm}
\end{equation}
Note that Eq.~(\ref{eq:onsitechemical}) allows for the addition of disordered on-site magnetic fields, see Ref.~\cite{yu2018beyond} for an example of $\eta$-pairing in such a setting. 
We only require that $\mE$ \textit{does not} depend on $\br$.
We now move on to the hopping term in Eq.~(\ref{eq:Hamilfinalmain}).
In App.~\ref{app:etareview}, we show the following (see Eq.~(\ref{eq:Tetacommreqd}))
\begin{eqnarray}
    &[\sumal{\sigma,\sigma'}{}{\hT^{\sigma,\sigma'}_{\br, \br'}}, q_{\br} \ed_{\br} + q_{\br'} \ed_{\br'}] = 0 \nn \\
\label{eq:commrelations}
\end{eqnarray}
provided $\{t^{\sigma,\sigma'}_{\br, \br'}\}$ and $\{q_{\br}\}$ satisfy (see Eq.~(\ref{eq:TSGAcondition}))
\begin{equation}
    q_{\br}t^{\sigma',\sigma}_{\br', \br}s_\sigma + q_{\br'}t^{\bar{\sigma},\bar{\sigma'}}_{\br, \br'}s_{\sigma'} = 0\;\;\;\forall \sigma,\sigma',\;\;\;\forall \langle \br, \br' \rangle,
\label{eq:SGAconditionmain}
\end{equation}
where we have defined
\begin{equation}
    s_\sigma \equiv \twopartdef{+1}{\sigma = \uparrow}{-1}{\sigma = \downarrow},\;\;\; \bar{\sigma} \equiv \twopartdef{\downarrow}{\sigma = \uparrow}{\uparrow}{\sigma = \downarrow}.
\label{eq:sigmadefns}
\end{equation}
Using Eqs.~(\ref{eq:etagraphMVcomm}) and (\ref{eq:commrelations}),  we obtain
\begin{equation}
    [\hgen, \ed] = \mE \ed,
\label{eq:etapairing}
\end{equation}
illustrating the generality of $\eta$-pairing.
\section{Examples of $\eta$-pairing}\label{sec:examples}
We now discuss a few examples of $\eta$-pairing with and without spin-orbit coupling.
%
%
\subsection{Without Spin-Orbit Coupling}\label{sec:withoutsoc}
We first consider the case without spin-orbit coupling. That is, we set
\begin{equation}
    t^{\uparrow, \downarrow}_{\br, \br'} = t^{\downarrow,\uparrow}_{\br, \br'} = 0,\;\;\;\forall \langle \br, \br' \rangle. 
\label{eq:noSOC}
\end{equation}
Eq.~(\ref{eq:SGAconditionmain}) reads
\begin{equation}
    q_{\br} t^{\sigma,\sigma}_{\br', \br} + q_{\br'} t^{\bar{\sigma}, \bar{\sigma}}_{\br,\br'} = 0\;\;\;\sigma \in \{\uparrow,\downarrow\},\;\;\forall \langle \br, \br'\rangle,
\end{equation}
leading to 
\begin{equation}
    \frac{q_{\br}}{q_{\br'}} = -\frac{t^{\downarrow,\downarrow}_{\br, \br'}}{(t^{\uparrow,\uparrow}_{\br,\br'})^\ast} = -\frac{t^{\uparrow,\uparrow}_{\br, \br'}}{(t^{\downarrow,\downarrow}_{\br,\br'})^\ast}.
\label{eq:noSOCconditions}
\end{equation}
Using Eq.~(\ref{eq:noSOCconditions}), we obtain
\begin{equation}
     |t^{\downarrow\downarrow}_{\br,\br'}| = |t^{\uparrow,\uparrow}_{\br,\br'}|  \implies \frac{|q_{\br}|}{|q_{\br'}|} = \frac{|t^{\downarrow,\downarrow}_{\br,\br'}|}{|t^{\uparrow,\uparrow}_{\br,\br'}|} = 1.
\label{eq:noSOCconditions1}
\end{equation}
Without loss of generality, as the norm of $q_{\br}$ is site-independent, we can set $|q_{\br}| = 1$ and choose
\begin{equation}
    q_{\br} = e^{i\phi_{\br}},\;\; t^{\uparrow, \uparrow}_{\br, \br'} = t_{\br, \br'} e^{i \theta^{\uparrow, \uparrow}_{\br, \br'}},\;\; t^{\downarrow, \downarrow}_{\br, \br'} = t_{\br, \br'} e^{i \theta^{\downarrow, \downarrow}_{\br, \br'}}    
\label{eq:phasechoice}
\end{equation}
where $t_{\br, \br'}$ is a positive real number, and
\begin{equation}
    \theta^{\uparrow,\uparrow}_{\br, \br'} + \theta^{\downarrow,\downarrow}_{\br, \br'} + \pi = \phi_{\br} - \phi_{\br'}.
\label{eq:phasecondition}
\end{equation}
We now illustrate some examples of an SGA in a disordered system.
For a one-dimensional chain of length $L$, $1 \leq \br \leq L$ the choice of hoppings 
\begin{equation}
\theta^{\uparrow,\uparrow}_{\br, \br'} = \theta^{\downarrow, \downarrow}_{\br, \br'} = 0
\label{eq:thetachoice}
\end{equation}
corresponds to the usual Hubbard model of Eq.~(\ref{eq:FermiHubbard}).
Thus, according to Eq.~(\ref{eq:phasecondition}), we see that we can choose $q_{\br} = e^{i\bp\cdot\br}$ when $L$ is even for periodic boundary conditions or any $L$ for open boundary conditions, recovering the standard $\ed$ operator of Eq.~(\ref{eq:etaops}).
In fact, for  bipartite graphs with sublattices $A$ and $B$, $\ed$ operators can be found for Hubbard models on by choosing
\begin{equation}
    q_{\br} = \twopartdef{+1}{\br \in A}{-1}{\br \in B},
\label{eq:qrbipartite}
\end{equation}
which has also been derived in Ref.~\cite{yang1992remarks}.
To obtain $\eta$-pairing states on non-bipartite lattices, we could choose $q_{\br} = \pm 1$, but we would necessarily have some pairs of nearest neighboring sites $\br$ and $\br'$ such that $\phi_{\br} = \phi_{\br'}$. 
Eq.~(\ref{eq:phasechoice}) for such $\br$ and $\br'$ can be satisfied by the choice 
\begin{equation}
    \theta^{\uparrow, \uparrow}_{\br, \br'} = \theta^{\downarrow, \downarrow}_{\br, \br'} = \frac{\pi}{2}.
\end{equation}
For example, on a triangular lattice, for every triangle with vertices denoted by $\br_1$, $\br_2$, $\br_3$ we could for example choose $\phi_{\br_1} = \phi_{\br_3} = 0$ and $\phi_{\br_2} = \pi$. In such a case, we need to choose for example $t_{\br_1, \br_2} = t_{\br_2, \br_3} = +t$ and $t_{\br_3, \br_1} = i t$ to satisfy Eq.~(\ref{eq:phasecondition}), which corresponds to the addition of a $\pi/2$ flux~\cite{zhai2005two}.
\subsection{With Spin-Orbit Coupling}\label{sec:withsoc}
We now explore the case when hopping terms with spin-orbit coupling are added to the generalized Hubbard Hamiltonian of Eq.~(\ref{eq:Hamilfinalmain}). 
In addition to Eq.~(\ref{eq:noSOC}), from Eq.~(\ref{eq:SGAconditionmain}) we obtain
\begin{equation}
    q_{\br} t^{\sigma,\bar{\sigma}}_{\br', \br} - q_{\br'} t^{\sigma, \bar{\sigma}}_{\br,\br'} = 0\;\;\;\sigma \in \{\uparrow,\downarrow\},\;\;\forall \langle \br, \br'\rangle, 
\end{equation}
enforcing that
\begin{equation}
    \frac{q_{\br}}{q_{\br'}} = \frac{t^{\uparrow,\downarrow}_{\br, \br'}}{\left(t^{\downarrow,\uparrow}_{\br,\br'}\right)^\ast} = \frac{t^{\downarrow,\uparrow}_{\br, \br'}}{\left(t^{\uparrow,\downarrow}_{\br,\br'}\right)^\ast} \implies |t^{\uparrow,\downarrow}_{\br, \br'}| = |t^{\downarrow,\uparrow}_{\br,\br'}|.
\label{eq:SOCconditions}
\end{equation}
Thus, in addition to Eq.~(\ref{eq:phasechoice}) we can choose
\begin{equation}
    t^{\uparrow, \downarrow}_{\br, \br'} = \widetilde{t}_{\br, \br'} e^{i \theta^{\uparrow, \downarrow}}_{\br, \br'},\;\;\;t^{\downarrow, \uparrow}_{\br, \br'} = \widetilde{t}_{\br, \br'} e^{i \theta^{\downarrow, \uparrow}}_{\br, \br'}
\label{eq:SOCphasechoice}
\end{equation}
where $\widetilde{t}_{\br, \br'}$ is a positive real number such that
\begin{equation}
    \theta^{\uparrow,\downarrow}_{\br, \br'} + \theta^{\downarrow, \uparrow}_{\br, \br'} = \phi_{\br} - \phi_{\br'}. 
\label{eq:SOCphasecondition}
\end{equation}
Thus, nearest-neighbor hopping terms with spin-orbit coupling to the Hubbard model on a bipartite graph (so that the hoppings are always between sites on different sublattices), provided they satisfy
\begin{equation}
    \theta^{\uparrow,\downarrow}_{\br, \br'} + \theta^{\downarrow,\uparrow}_{\br, \br'} = \pm \pi \;\; \implies t^{\uparrow, \downarrow}_{\br, \br'} = -t^{\downarrow, \uparrow}_{\br, \br'},  
\end{equation}
where we have used Eqs.~(\ref{eq:phasecondition}), (\ref{eq:thetachoice}),  and (\ref{eq:SOCphasecondition}).
We can indeed verify that in this limit, we recover the conditions derived for $\eta$-pairing in translation-invariant spin-orbit coupled Hubbard models in Ref.~\cite{li2019exact}.
\section{Quantum Many-Body Scars from the Hubbard Model}\label{sec:towers}
As we showed in the previous sections, the existence of an SGA in the generalized Hubbard models gives rise to several towers of $\eta$-pairing states.
We now ask if perturbations can be added to those models that preserve  \textit{some} but not all of the towers generated by $\eta$-pairing are preserved. 
We introduce the concept of a Restricted Spectrum Generating Algebra (RSGA), a restriction of the SGA discussed in Sec.~\ref{sec:etareview}, and illustrate perturbations that realize those conditions.
These perturbed Hamiltonians hence preserve some towers generated by $\eta$-pairing, and we argue that the resulting towers of eigenstates become quantum many-body scars in the perturbed Hamiltonians.
We note that everything we derive here will apply to both the original Hubbard model of Eq.~(\ref{eq:HubbardPBC}) and the generalized Hubbard models of Eq.~(\ref{eq:Hamilfinalmain}), but we focus on the latter for the sake of generality.  
\subsection{RSGA of Order 1}\label{sec:RSGA1}
A Hamiltonian $H$ is said to exhibit a \textit{Restricted Spectrum Generating Algebra of Order 1} (RSGA-1) if there exists a state $\ket{\psi_0}$ and an operator $\ed$ such that $\ed\ket{\psi_0} \neq 0$ that satisfy 
\begin{eqnarray}
    &&\textrm{(i)}\ H\ket{\psi_0} = E_0 \ket{\psi_0} \nn \\
    &&\textrm{(ii)}\ [H, \ed]\ket{\psi_0} = \mE\ed\ket{\psi_0} \nn \\
    &&\textrm{(iii)}\ [[H, \ed], \ed] = 0.
\label{eq:RSGA1cond}
\end{eqnarray}
As we show in Lemma~\ref{lem:RSGA1} in App.~\ref{app:RSGAtower}, the conditions of Eq.~(\ref{eq:RSGA1cond}) lead to the existence of a equally-spaced tower of states $\{(\ed)^n \ket{\psi_0}\}$ starting from $\ket{\psi_0}$.  
We illustrate this concept by choosing $\ket{\psi_0} = \ket{\Omega}$, the empty vacuum state, and as a perturbation of the Hamiltonian  $\hgen$, the electrostatic interaction of the form
\begin{equation}
    \hI_2 \equiv \sumal{\sigma,\sigma'}{}{\sumal{\dlangle \br, \br'\drangle}{}{V^{\sigma, \sigma'}_{\br,\br'} \hn_{\br, \sigma} \hn_{\br', \sigma'}}},
\label{eq:Int}
\end{equation}
where $\dlangle\br,\br'\drangle$ runs over some or all pairs of sites on the graph. 
Note that this sum can be restricted to only nearest-neighbor sites for a more physical choice of interaction. 
Since $\ket{\Omega}$ is an eigenstate of the Hubbard Hamiltonian $\hgen$ and $\hI_2$, we obtain
\begin{equation}
    (\hgen + \hI_2)\ket{\Omega} = 0,
\end{equation}
satisfying condition (i) of RSGA-1 with $E_0 = 0$. 
Further, using the commutation relation in Eq.~(\ref{eq:usefuletacoms1}) 
\begin{equation}
    [\hn_{\br, \sigma}, \ed_{\br}] = \ed_{\br}, 
\label{eq:numbercomm}
\end{equation}
we deduce for $\br \neq \br'$ that
\begin{equation}
    [\hn_{\br, \sigma} \hn_{\br', \sigma'}, q_{\br} \ed_{\br} + q_{\br'}\ed_{\br'}] = q_{\br}\ed_{\br}\hn_{\br',\sigma'} + q_{\br'} \hn_{\br, \sigma}\ed_{\br'}.
\label{eq:electrocomm}
\end{equation}
Using Eq.~(\ref{eq:electrocomm}), we note that
\begin{equation}
    [\hn_{\br, \sigma} \hn_{\br',\sigma'}, \ed]\ket{\Omega} = [\hn_{\br, \sigma} \hn_{\br',\sigma'}, q_{\br} \ed_{\br} + q_{\br'}\ed_{\br'}]\ket{\Omega} = 0. 
\label{eq:vanishingcond}
\end{equation} 
due to the $\br$ and $\br'$ occupations of the vacuum state.
As a consequence, we obtain $[\hI_2, \ed]\ket{\Omega} = 0$ and $[\hgen + \hI_2, \ed] = \mE\ed$, satisfying condition (ii) of RSGA-1. 
Further, we note that using Eqs.~(\ref{eq:numbercomm}) and (\ref{eq:electrocomm}), we obtain
\begin{equation}
    [[\hn_{\br, \sigma} \hn_{\br', \sigma'}, q_{\br}\ed_{\br} + q_{\br'}\ed_{\br'}], q_{\br}\ed_{\br} + q_{\br'}\ed_{\br'}] = 2 q_{\br} q_{\br'} \ed_{\br} \ed_{\br'}.
\label{eq:electrocomm2}
\end{equation}
Thus, we obtain
\begin{equation}
    [[\hI_2, \ed], \ed] = 2\sumal{\dlangle \br, \br'\drangle}{}{(\sumal{\sigma, \sigma'}{}{V^{\sigma, \sigma'}_{\br, \br'}})q_{\br}q_{\br'}\ed_{\br} \ed_{\br'}}.
\label{eq:I2doubcomm}
\end{equation}
Setting
\begin{equation}
    \sumal{\sigma, \sigma'}{}{V^{\sigma, \sigma'}_{\br, \br'}} = 0,
\label{eq:sumVcond}
\end{equation}
and using Eqs.~(\ref{eq:I2doubcomm}) and (\ref{eq:etapairing}), we obtain
\begin{equation}
    [[\hgen + \hI_2, \ed],\ed] = 0,
\end{equation}
satisfying condition (iii) of RSGA-1.
Thus, as a consequence of Lemma~\ref{lem:RSGA1}, the Hamiltonian $(\hgen + \hI_2)$ exhibits the tower $\{(\ed)^n\ket{\Omega}\}$ as eigenstates, although other $\eta$-pairing towers (starting from other states than the vacuum state) of the Hubbard Hamiltonian might not be preserved.   
A simple physical interaction that satisfies Eq.~(\ref{eq:sumVcond}) is the nearest neighbor $S_z-S_z$ interaction.
\subsection{RSGA of Order $\boldsymbol{M}$}\label{sec:RSGAM}
We now study perturbations to the Hubbard model that do not satisfy the conditions of RSGA-1 but still preserve a tower of states.
The concept of RSGA-1 can be generalized straightforwardly as follows.
We define a set of states $\{\ket{\psi_n}\}$ as
\begin{equation}
    \ket{\psi_n} \equiv (\ed)^n\ket{\psi_0}.
\label{eq:psindefn}
\end{equation}
We define a set of operators $\{H_n\}$ as
\begin{equation}
    H_0 \equiv H,\;\; H_{n+1} \equiv [H_n, \ed],\;\;\forall n \geq 0. 
\label{eq:Hndefn}
\end{equation} 
A Hamiltonian $H$ is said to exhibit a \textit{Restricted Spectrum Generating Algebra of Order $M$} (RSGA-$M$) if there exists a state $\ket{\psi_0}$ and an operator $\ed$ such that $\ket{\psi_n} \neq 0$ for $n \leq M$ that satisfy
\begin{eqnarray}
    &&\textrm{(i)}\ H\ket{\psi_0} = E_0\ket{\psi_0} \nn \\
    &&\textrm{(ii)}\ H_1\ket{\psi_0} = \mE\ed\ket{\psi_0}\nn \\
    &&\textrm{(iii)}\ H_n\ket{\psi_0} = 0 \;\; \forall\ n,\ 2 \leq n \leq M\nn \\
    &&\textrm{(iv)}\ H_n \twopartdef{\neq 0}{n \leq M}{= 0}{n = M+1},
\label{eq:RSGA2cond}
\end{eqnarray}
where condition (iii) of RSGA-1 of Eq.~(\ref{eq:RSGA1cond}) has been modified.
As we show in Lemma~\ref{lem:RSGA2} in App.~\ref{app:RSGAtower}, the conditions of Eq.~(\ref{eq:RSGA2cond}) is equivalent to the existence of a equally-spaced tower of states $\{(\ed)^n \ket{\psi_0}\}$ starting from $\ket{\psi_0}$.  
Note that conditions (i)-(iii) of RSGA-$M$ lead to the existence of exact eigenstates $\{\ket{\psi_0}, \ed\ket{\psi_0}, \cdots, (\ed)^M\ket{\psi_0}\}$ with energies $\{E_0, E_0 + \mE, \cdots, E_0 + M \mE\}$. 
If we do not add condition (iv), then these are all the guaranteed eigenstates, for a given $M$. 
Condition (iv) ensures that $(\ed)^n\ket{\psi_0}$ is also an eigenstate of $H$ for any $n$ as long as it does not vanish.
We now explicitly construct a perturbation to the generalized Hubbard model $\hgen$ that admits an RSGA of order $M$.
Consider the $(M+1)$-body density interaction term
\begin{equation}
	\hI_{M+1} \equiv \sumal{\{\br_j\}}{}{V^{\{\sigma_j\}}_{\{\br_j\}}\prodal{j = 1}{M+1}{\hn_{\br_j, \sigma_j}}},
\label{eq:InteractionN}
\end{equation}
where $\{\br_j\}$ represent a (chosen) set of $(M+1)$ distinct sites and $\{\sigma_j\}$ a set of $(M+1)$ spins. 
On the vacuum state $\ket{\Omega}$, we obtain
\begin{equation}
    (\hgen + \hI_{M+1})\ket{\Omega} = 0,
\end{equation}
satisfying condition (i) of RSGA-$M$ with $E_0 = 0$. 
Using Eq.~(\ref{eq:numbercomm}), we obtain
\begin{equation}
    [\prodal{j = 1}{M+1}{\hn_{\br_j, \sigma_j}}, \ed] = \sumal{k = 1}{M+1}{q_{\br_k} \ed_{\br_k} \prodal{j = 1, j\neq k}{M+1}{\hn_{\br_j, \sigma_j}}},
\end{equation}
and thus $[\hI_{M+1}, \ed]\ket{\Omega} = 0$. 
Using Eq.~(\ref{eq:etapairing}) we further obtain $[\hgen + \hI_{M+1}, \ed]\ket{\Omega} = \mE\ed\ket{\Omega}$, satisfying condition (ii) of RSGA-$M$.
Similarly, we can compute subsequent commutators with $\ed$. 
Since according to Eq.~(\ref{eq:numbercomm}) each commutator replaces an $\hn_{\br, \sigma}$ by $\ed_{\br}$, applying less than $(M+1)$ commutators, we obtain
\begin{equation}
    \underbrace{[[\prodal{j = 1}{M+1}{\hn_{\br_j, \sigma_j}}, \ed], \ed] \cdots ]\cdots]]}_{n\ \textrm{times}} \neq 0 \;\; \forall\ n,\ 2 \leq n \leq M.
\end{equation}
Furthermore, since the commutator applied less than $(M+1)$ times necessarily consists of at least one number operators $\hn_{\br, \sigma}$ in each term, it vanishes on the vacuum state $\ket{\Omega}$, i.e. 
\begin{equation}
    \underbrace{[[\prodal{j = 1}{M+1}{\hn_{\br_j, \sigma_j}}, \ed], \ed] \cdots ]\cdots]]}_{n\ \textrm{times}}\ket{\Omega} = 0 \;\; \forall\ n,\ 2 \leq n \leq M.
\end{equation}
The interaction $\hI_{M+1}$ of Eq.~(\ref{eq:InteractionN}) along with Eq.~(\ref{eq:etapairing}) thus satisfies conditions (iii) of Eq.~(\ref{eq:RSGA2cond}) with $\ket{\psi_0} = \ket{\Omega}$.  
Applying the commutator $(M+1)$ times, we obtain
\begin{equation}
	\underbrace{[[\prodal{j = 1}{M+1}{\hn_{\br_j, \sigma_j}}, \ed], \ed] \cdots ]\cdots]]}_{M+1\ \textrm{times}} = \prodal{j = 1}{M+1}{\ed_{\br_j}}.
\label{eq:prodcomm}
\end{equation}
Using Eqs.~(\ref{eq:InteractionN}) and (\ref{eq:prodcomm}), we obtain
\begin{equation}
    \underbrace{[[\hI_{M+1}, \ed], \ed] \cdots ]\cdots]]}_{M+1\ \textrm{times}} = \left(\sumal{\{\br_j\}, \{\sigma_j\}}{}{V^{\{\sigma_j\}}_{\{\br_j\}}}\right)\prodal{j = 1}{M}{\ed_{\br_j}}.
\end{equation}
Condition (iv) of Eq.~(\ref{eq:RSGA2cond}) can be satisfied if 
\begin{equation}
    \sumal{\{\br_j\}, \{\sigma_j\}}{}{V^{\{\sigma_j\}}_{\{\br_j\}}} = 0,
\label{eq:Imcondition}
\end{equation}
and the Hamiltonian $(\hgen + \hI_{M+1})$ exhibits an RSGA of order $M$.
Note that while these perturbations preserve the same tower of states $\{(\ed)^n\ket{\Omega}\}$ as the perturbations illustrated in Sec.~\ref{sec:RSGA1}, the algebra is different; this allows us to obtain many different terms that can be added to the Hamiltonian in order to maintain this tower of states.
\subsection{Connections to Quantum Scars}\label{sec:quantumscarconnection}
We now prove that the towers $\{(\ed)^n \ket{\Omega}\}$ in the Hamiltonians $(\hhub + \hI_{M+1})$ discussed as examples in Secs.~\ref{sec:RSGA1} and \ref{sec:RSGAM} are the quantum many-body scars for appropriate values of the Hamiltonian parameters, when it is non-integrable and expected to satisfy ETH. 
Physically, the states of the tower are composed of doubly-occupied quasiparticles ``doublons" dispersing on top of the vacuum state $\ket{\Omega}$.
The energy of a doublon is $(U - 2\mu)$ under the Hubbard Hamiltonian, hence the state $(\ed)^n\ket{\Omega}$ has an energy $n(U - 2\mu)$, spin $S_z = 0$, and charge $Q = 2n$. 
The highest state of the tower consists of all the sites being filled with doublons.
As expected for quasiparticles on top of a low entanglement state~\cite{Moudgalya2018b, schecter2019weak}, and as rigorously computed in Ref.~\cite{vafek2017entanglement}, the entanglement entropy $\mS$ of the state $(\ed)^n \ket{\Omega}$ scales as the logarithm of the subsystem volume ($\mS \sim \log V$), in contrast to the volume-law ($\mS \sim V$) predicted by ETH~\cite{d2016quantum}. 
By estimating the energies of the states in the sector with spin $S_z = 0$ and charge $Q = 2n$, we now show that some states of the tower $\{(\ed)^n \ket{\Omega}\}$ can be in the bulk of the spectrum of their quantum number sectors. Note that, unlike in  Ref.~\cite{vafek2017entanglement}, we have now lost the $\eta$-pairing symmetry, and hence the theorem, proven in  Ref.~\cite{vafek2017entanglement} - that the $\eta$-pairing states are the only ones in their quantum number sectors, does not apply. 
For example, consider the Hubbard model $\hhub$ of Eq.~(\ref{eq:FermiHubbard}) in one dimension with an even system size $L$.
Non-interacting ferromagnetic eigenstates with charge $Q = 2n$ can be constructed by occupying the single-particle spectrum of the quadratic part of the Hamiltonian $\hhub$ with $2n$ $\uparrow$ spins. 
These eigenstates have spin quantum numbers $S_z = n$, and the lowest and highest energies $E_-$ and $E_+$ of such states are given by
\begin{eqnarray}
    E_{\pm} &=& -2n\mu\ \pm\ 2t\ \sumal{j = -n}{n-1}{\cos\left(\frac{2\pi j}{L}\right)} \nn \\
    &=&-2n\mu\ \pm\ 2t\ \cot\left(\frac{\pi}{L}\right)\sin\left(\frac{2 n \pi}{L}\right).
\label{eq:lowhigheta}
\end{eqnarray}
As a consequence of the spin-$SU(2)$ symmetry of $\hhub$, eigenstates with the same energies but with spin $S_z = 0$ can be constructed by applying the spin lowering operator on these non-interacting ferromagnetic eigenstates.
By adding a perturbation $\hI_{M+1}$ that breaks spin-flip symmetry and translation invariance to $\hhub$, we break the integrability of the one-dimensional Hubbard model~\footnote{We have numerically checked that the Hamiltonians $(\hhub + \hI_{2})$ with nearest-neighbor electrostatic interactions exhibit level repulsion and GOE level statistics for generic values of couplings $\{V^{\sigma, \sigma'}_{\br, \br'}\}$, even if they satisfy Eq.~(\ref{eq:sumVcond}).} and all the symmetries of $\hhub$ except spin $S_z$ and charge $U(1)$.
For a small perturbation strength, we expect the lowest and highest energy eigenstates of the $S_z = 0$ and $Q = 2n$ sector to still be upper and lower bounded by (approximately) $E_-$ and $E_+$ respectively. 
Thus, we expect the state $(\ed)^n\ket{\Omega}$ to be certainly in the bulk of the spectrum of its own quantum number sector $S_z = 0$, $Q = 2n$ if $E_-  < n (U - 2 \mu) < E_+$, or,
\begin{equation}
-\frac{2t}{n}\cot\left(\frac{2\pi}{L}\right)\sin\left(\frac{2n \pi}{L}\right) < U < \frac{2t}{n}\cot\left(\frac{2\pi}{L}\right)\sin\left(\frac{2n \pi}{L}\right),
\label{eq:bulkcondition}
\end{equation}
which can always be satisfied by an appropriate choice of $U$ and $t$. 
For a finite density of doublons in the thermodynamic limit ($n/L = \rho$ and $n, L \rightarrow \infty$), using Eq.~(\ref{eq:bulkcondition}) we obtain
\begin{equation}
    -\frac{\sin\left(2\pi \rho\right)}{\pi\rho} < \frac{U}{t} < \frac{\sin\left(2\pi\rho\right)}{\pi\rho}.    
\label{eq:bulkconditionrho}
\end{equation}
We could also add small spin-orbit coupling and disorder to $\hhub$ to obtain $\hgen$.
This breaks the spin $U(1)$ symmetry, which combines the quantum number sectors of various $S_z$'s with the same $Q$.  
The estimate of Eq.~(\ref{eq:bulkconditionrho}) is thus a condition under which some states of the tower $\{(\ed)^n\ket{\Omega}\}$ are quantum many-body scars of the Hamiltonian $(\hgen + \hI_{M+1})$.  
While we have broken translation invariance here, in App.~\ref{app:scarsTI} we show that these scars are in the bulk of the spectrum as long as Eq.~(\ref{eq:bulkconditionrho}) is satisfied, even when translation, inversion, and spin-flip symmetries are not broken. 
Further, in App.~\ref{app:scarshigherD}, we obtain similar conditions for the towers of states in $D$-dimensional models to lie in the bulk of the spectrum.
\section{RSGA and Quantum Scarred Models}\label{sec:RSGAquantumscars}
Exact towers of states as discussed in Sec.~\ref{sec:towers} are also found in several models of exact quantum many-body scars~\cite{Moudgalya2018a, Moudgalya2018b, schecter2019weak, iadecola2019quantum2, chattopadhyay2019quantum, mark2020unified, moudgalya2020large}. 
In this section, we briefly comment on the connections between the RSGA formalism introduced here and quantum scarred models in the literature, and in particular the unified formalism introduced in Ref.~\cite{mark2020unified}.
The theorem of Eq.~(1) in Ref.~\cite{mark2020unified} states that given an eigenstate $\ket{\psi_0}$ of Hamiltonian $H$ with energy $E_0$, and a subspace $W$ such that $\ket{\psi_0} \in W$, a tower of equally spaced states $\{(\ed)^n\ket{\psi_0}\}$ with energies $\{E_0 + n\mE\}$ is guaranteed if for any $\ket{\psi} \in W$
\begin{equation}
    (i)\ [H, \ed] \ket{\psi} = \mE \ed \ket{\psi},\;\; (ii)\ \ed \ket{\psi} \in W.
\label{eq:MLMcondition}
\end{equation}
Since the RSGAs guarantee the existence of a tower of states $\{(\ed)^n\ket{\psi_0}\}$, they satisfy the conditions of Eq.~(\ref{eq:MLMcondition}) by choosing the subspace $W = \textrm{span}\{\ket{\psi_0}, \ed\ket{\psi_0}, \cdots, (\ed)^N\ket{\psi_0}\}$, and are captured by the formalism of Ref.~\cite{mark2020unified}.
However, since we can obtain RSGAs of any order, they provide a \textit{finer} classification of quantum scarred models.
We illustrate this connection by focusing on two examples: (i) the spin-1 XY model family studied in Ref.~\cite{schecter2019weak}, which we find admit RSGAs of order $M = 1$, and (ii) the families of spin-1 scarred Hamiltonians (including the AKLT Hamiltonian) studied in Refs.~\cite{Moudgalya2018a, mark2020unified, moudgalya2020large}, which we find admit RSGAs of order $M = 2$.
For pedagogical purposes, we now detail the former and provide a similar derivation for the latter in App.~\ref{app:RSGAscars}.
The spin-1 XY Hamiltonian family on a $D$-dimensional hypercubic lattice with size $L$ in each dimension is given by
\begin{equation}
    \hxy = \underbrace{J\sumal{\langle \br, \br'\rangle}{}{(S^x_{\br} S^x_{\br'} + S^y_{\br} S^y_{\br'})}}_{H_{XY}} + \underbrace{h\sumal{\br}{}{S^z_{\br}}}_{H_z} + \underbrace{D\sumal{\br}{}{(S^z_{\br})^2}}_{H_{z^2}}. 
\label{eq:XYHamil}
\end{equation}
Throughout this section we label the spin-1 degrees of freedom by $\{+, 0, -\}$.
As discussed in Ref.~\cite{schecter2019weak}, the spin-1 XY Hamiltonian has a tower of states starting from a spin-polarized root eigenstate $\ket{\Omega} \equiv \ket{-\ -\ \cdots\ -\ -}$:
\begin{equation}
    \hxy (\pxy)^n \ket{\Omega} = (h(2n - L^D) + D L^D)(\pxy)^n \ket{\Omega},
\label{eq:RSGA0XY}
\end{equation}
for $0 \leq n \leq L^D$ and
\begin{equation}
    \pxy \equiv  \sumal{\br}{}{e^{i \pi\cdot \br} (S^+_{\br})^2}.
\label{eq:raisXY}
\end{equation}
Ref.~\cite{mark2020unified} showed that
\begin{eqnarray}
    &[H_z, \pxy] = 2h\pxy,\;\;[H_{z^2}, \pxy] = 0,\nn \\
    &[H_{XY}, \pxy] = 4J\sumal{\langle \br, \br'\rangle}{}{e^{i \bp\cdot \br} \hh_{\br,\br'}},
\label{eq:XYcomms}
\end{eqnarray}
where $\hh_{\br,\br'}$ reads
\begin{equation}
    \hh_{\br,\br'} = \ket{0\ +}\bra{-\ 0} - \ket{+\ 0}\bra{0\ -}.
\label{eq:hjk}
\end{equation}
We can thus decompose $\hxy$ as 
\begin{equation}
	\hxy = \underbrace{H_{z^2} + H_{z}}_{\hxy_{\textrm{SGA}}}+ \underbrace{H_{XY}}_{\vxy}, 
\label{eq:HXYdecomp}
\end{equation}
where $\hxy_{\textrm{SGA}}$ admits an exact SGA, i.e.,
\begin{equation}
	[\hxy_{\textrm{SGA}}, \pxy] = \exy\pxy,\;\;\;\exy = 2h.
\label{eq:HXYSGA}
\end{equation}
We thus obtain
\begin{equation}
    [\hxy, \pxy] = 2h\pxy + 4J\sumal{\langle \br, \br'\rangle}{}{e^{i\bp\cdot \br} \hh_{\br,\br'}}.
\label{eq:XYcomm}
\end{equation}
Noting that $\hh_{\br,\br'}\ket{\Omega} = 0$, we obtain
\begin{equation}
    [\hxy, \pxy]\ket{\Omega} = 2h\pxy\ket{\Omega}. 
\label{eq:RSGA1XY}
\end{equation}
Further, using Eqs.~(\ref{eq:XYcomm}) and (\ref{eq:hjk}), we obtain
\begin{equation}
    [[\hxy, \pxy], \pxy] = 4J\sumal{\langle \br, \br' \rangle}{}{[\hh_{\br,\br'}, (S^+_{\br})^2 - (S^+_{\br'})^2]} = 0. 
\label{eq:RSGA2XY}
\end{equation}
Using Eqs.~(\ref{eq:RSGA0XY}), (\ref{eq:RSGA1XY}), and (\ref{eq:RSGA2XY}), we obtain that the family of Hamiltonians of Eq.~(\ref{eq:RSGA0XY}) admit an RSGA of order $M = 1$ (see Lemma~\ref{lem:RSGA1} in App.~\ref{app:RSGAtower}) with $\ket{\psi_0} = \ket{\Omega}$, $E_0 = (D - h) L^D$, $\mE = 2h$, $\ed = \pxy$. 
%
%
%

%
%
%
%
%
%
%
%
\section{Conclusions}\label{sec:conclusions}
In this article, we have shown how quantum many-body scars based on $\eta$-pairing states can appear in generalized and perturbed fermionic Hubbard models. 
We have explored the $\eta$-pairing states in the Hubbard model and generalized in two directions.  
First, casting $\eta$-pairing as a real-space phenomenon, we find a highly general Hubbard Hamiltonian potentially with disorder and spin-orbit coupling that exhibits a Spectrum Generating Algebra (SGA). 
Second, we introduce the concept of Restricted SGA (RSGA) and add use it to find various perturbations to the (generalized) Hubbard models that preserve the $\eta$-pairing tower starting from the vacuum state.  
The states of this tower have a sub-thermal entanglement entropy, and we analytically obtain conditions for the states of this tower to lie in the bulk of the spectrum of their quantum number sector, showing that they are examples of quantum many-body scars. 
We further connected RSGAs to some models of exactly solvable quantum scars in the literature, particularly the first two examples of towers of quantum scars in the AKLT model~\cite{Moudgalya2018a} and the spin-1 XY model~\cite{schecter2019weak}.
The scars there can thus be explained by the existence of RSGAs obtained by perturbing Hamiltonians with exact SGAs.
There are many natural extensions to this work.
It is important to understand the connection of RSGAs with models of quantum scars that exhibit multi-site quasiparticles~\cite{iadecola2019quantum2, chattopadhyay2019quantum, moudgalya2020large}, including Hubbard models with generalized $\eta$-pairing~\cite{vafek2017entanglement}, where the discussions in this work do not seem to generalize easily. 
Further, the RSGAs described here closely resemble algebraic structures introduced in earlier works both in the context of ground states~\cite{batista2009canted, wouters2018exact} as well as quantum scars~\cite{mark2020unified, bull2020quantum}, and it is highly desirable to better understand the connections between them, and also connections to the embedding construction in Ref.~\cite{shiraishi2019connection}. 
Appropriate generalizations of RSGAs might also provide a way to construct closed solvable subspaces that are not necessarily equally spaced towers of states, akin to the closed Krylov subspaces found in several constrained systems~\cite{iadecola2018exact, sala2019ergodicity, khemani2019local, moudgalya2019thermalization}.
On a different note, given that the SGAs and RSGAs survive in the presence of disorder, it would be interesting to understand the existence and implications of these towers of states in the many-body localized regime in Hubbard models~\cite{mondaini2015many, prelov2016absence, kondov2015disorder, kozarzewski2018spin, mierzejewski2018counting}.
Beyond Hamiltonian systems, it would be interesting to explore $\eta$-pairing in Floquet~\cite{kitamura2016eta} and open quantum systems~\cite{buca2019non, tindall2019heating}, and obtain RSGA-like algebraic structures to construct models of Floquet quantum many-body scars~\cite{pai2019robust, mukherjee2019collapse, haldar2019scars, sugiura2019manybody, zhao2020quantum, mizuta2020exact}.
\textit{Note added}: A related work by D.~K.~Mark and O.~I.~Motrunich~\cite{mark2020etapairing} appeared in the same arXiv posting.
\section*{Acknowledgements}
We thank Abhinav Prem and Lesik Motrunich for useful discussions, and Hosho Katsura and Daniel Mark for comments on a draft.  
S.M. acknowledges the hospitality of LPENS, Paris and NORDITA, Stockholm where parts of this work were completed.
N.R and B.A.B. were supported by the Department of Energy Grant No. de-sc0016239, the Schmidt Fund for Innovative Research, Simons Investigator Grant No. 404513  the Packard Foundation. Further support was provided by the National Science Foundation EAGER Grant No. DMR 1643312, NSF-MRSEC DMR-1420541, BSF Israel US foundation No. 2018226, and ONR No. N00014-20-1-2303. 
\appendix
\section{Useful Identities}\label{app:fermialgebra}
In this appendix, we provide some useful operator identities that we use in this article.
We denote spinful fermionic creation and annihilation operators by $\{\cd_{\br,\sigma}\}$ and $\{c_{\br,\sigma}\}$, where $\br$ denotes the site index and $\sigma$ the spin index.
These obey the algebra
\begin{eqnarray}
    &\{c_{\br, \sigma}, c_{\br', \sigma'}\} = \{\cd_{\br, \sigma}, \cd_{\br', \sigma'}\} = 0 \nn \\
    &\{c_{\br, \sigma}, \cd_{\br', \sigma'}\} = \delta_{\br, \br'} \delta_{\sigma, \sigma'}. 
\label{eq:fermialgebra}
\end{eqnarray}
Further, defining the operators
\begin{equation}
   \hn_{\br,\sigma} \equiv \cd_{\br, \sigma} c_{\br, \sigma},\;\;\; \ed_{\br} \equiv \cd_{\br, \uparrow} \cd_{\br, \downarrow},\;\;\; \eta_{\br} \equiv -c_{\br, \uparrow} c_{\br, \downarrow}
\label{eq:numetadefn}
\end{equation}
we directly obtain the useful relations
\begin{eqnarray}
	&&[\cd_{\br, \sigma}, \hn_{\br', \sigma'}] = - \delta_{\br, \br'} \delta_{\sigma, \sigma'} \cd_{\br, \sigma},\nn \\
	&&[c_{\br, \sigma}, \hn_{\br', \sigma'}] = \delta_{\br, \br'} \delta_{\sigma, \sigma'} c_{\br, \sigma}.
\label{eq:usefulcoms}
\end{eqnarray}
We also obtain
\begin{eqnarray}
    &&[\hn_{\br, \sigma}, \ed_{\br'}] = \delta_{\br, \br'} \ed_{\br} ,\nn\\
    &&[\hn_{\br, \uparrow} \hn_{\br,\downarrow}, \ed_{\br'}] = \delta_{\br, \br'} \ed_{\br},\label{eq:usefuletacoms1} \\
    &&[\ed_{\br'}, \ed_{\br}] = 0, \nn \\
    &&[\eta_{\br'}, \ed_{\br} ] = \delta_{\br, \br'}(1 - \hn_{\br, \uparrow} - \hn_{\br, \downarrow}).
\label{eq:usefuletacoms2}
 \end{eqnarray}
Using Eq.~(\ref{eq:usefulcoms}), we also obtain
\begin{equation}
    [\cd_{\br',\sigma'} c_{\br,\sigma}, \ed_{\br}] = -s_\sigma \cd_{\br,\bar{\sigma}} \cd_{\br',\sigma'},
\label{eq:simpcom}
\end{equation}
where we have defined
\begin{equation}
    s_\sigma = \twopartdef{+1}{\sigma = \uparrow}{-1}{\sigma = \downarrow},\;\;\; \bar{\sigma} = \twopartdef{\downarrow}{\sigma = \uparrow}{\uparrow}{\sigma = \downarrow}.
\label{eq:sigmadefnsapp}
\end{equation}

\section{$\boldsymbol{\eta}$-pairing with disorder and spin-orbit coupling}\label{app:etareview}
Here we derive the conditions for the $\eta$ operator of Eq.~(\ref{eq:etadefn}) to commute with a generic one-body hopping operator of the form
\begin{equation}
    \hT^{\sigma, \sigma'}_{\br, \br'} = \left(t^{\sigma,\sigma'}_{ \br, \br'}\cd_{\br, \sigma} c_{\br', \sigma'} + t^{\sigma',\sigma}_{\br', \br}\cd_{\br', \sigma'} c_{\br, \sigma}\right).
\label{eq:onebodyhopping}
\end{equation}  
Our aim is to determine a set of conditions on $\{q_{\br}\}$, $\{t^{\sigma,\sigma'}_{\br,\br'}\}$ such that
\begin{equation}
    [\sumal{\sigma,\sigma'}{}{\hT^{\sigma, \sigma'}_{\br, \br'}}, \ed] = [\sumal{\sigma,\sigma'}{}{\hT^{\sigma, \sigma'}_{\br, \br'}}, q_{\br} \ed_{\br} + q_{\br'}\ed_{\br'}] = 0.
\label{eq:Tetacommreqd}
\end{equation}
We first compute $[\hT^{\sigma,\sigma'}_{\br,\br'}, q_{\br} \ed_{\br}]$:
\begin{eqnarray}
    [\hT^{\sigma, \sigma'}_{\br, \br'}, q_{\br} \ed_{\br}] &=& [t^{\sigma,\sigma'}_{\br, \br'}\cd_{\br, \sigma} c_{\br', \sigma'} + t^{\sigma',\sigma}_{\br', \br} \cd_{\br', \sigma'} c_{\br, \sigma},\;\; q_{\br}\cd_{\br,\uparrow}\cd_{\br,\downarrow}] \nn \\ 
    &=& [t^{\sigma',\sigma}_{\br', \br} \cd_{\br', \sigma'} c_{\br, \sigma},\;\; q_{\br}\cd_{\br,\uparrow}\cd_{\br,\downarrow}] \nn \\
    &=& -q_{\br} t^{\sigma',\sigma}_{\br',\br} s_{\sigma} \cd_{\br, \bar{\sigma}}\cd_{\br',\sigma'},
\label{eq:Tetarcomm}
\end{eqnarray}
where we have used Eqs.~(\ref{eq:fermialgebra}) and (\ref{eq:simpcom}).
Similarly, we also obtain
\begin{eqnarray}
    [\hT^{\sigma, \sigma'}_{\br, \br'}, q_{\br'} \ed_{\br'}] &=& [t^{\sigma,\sigma'}_{\br, \br'}\cd_{\br, \sigma} c_{\br', \sigma'} + t^{\sigma',\sigma}_{\br', \br} \cd_{\br', \sigma'} c_{\br, \sigma},\;\; q_{\br'}\cd_{\br',\uparrow}\cd_{\br',\downarrow}] \nn \\ 
    &=& [t^{\sigma,\sigma'}_{\br, \br'}\cd_{\br, \sigma} c_{\br', \sigma'},\;\; q_{\br'}\cd_{\br',\uparrow}\cd_{\br',\downarrow}] \nn \\
    &=& -q_{\br'} t^{\sigma,\sigma'}_{\br,\br'} s_{\sigma'} \cd_{\br', \bar{\sigma'}}\cd_{\br,\sigma}.
\label{eq:Tetarpcomm}
\end{eqnarray}
Using Eqs.~(\ref{eq:Tetacommreqd}), (\ref{eq:Tetarcomm}), and (\ref{eq:Tetarpcomm}), we obtain
\begin{eqnarray}
    [\sumal{\sigma,\sigma'}{}{\hT^{\sigma,\sigma'}_{\br, \br'}}, \ed] 
    &=& -\sumal{\sigma,\sigma'}{}{\left(q_{\br}t^{\sigma',\sigma}_{\br', \br}s_\sigma - q_{\br'}t^{\bar{\sigma},\bar{\sigma'}}_{\br, \br'}s_{\bar{\sigma'}}\right)\cd_{\br,\bar{\sigma}} \cd_{\br',\sigma'}} \nn \\
    &=& -\sumal{\sigma,\sigma'}{}{\left(q_{\br}t^{\sigma',\sigma}_{\br', \br}s_\sigma + q_{\br'}t^{\bar{\sigma},\bar{\sigma'}}_{\br, \br'}s_{\sigma'}\right)\cd_{\br,\bar{\sigma}} \cd_{\br',\sigma'}},\nn \\
\label{eq:etagraphTcomm}
\end{eqnarray}
where we have used Eq.~(\ref{eq:fermialgebra}) and $s_{\bar{\sigma'}} = -s_{\sigma'}$.
Eq.~(\ref{eq:Tetacommreqd}) is thus satisfied by setting 
\begin{equation}
    q_{\br}t^{\sigma',\sigma}_{\br', \br}s_\sigma + q_{\br'}t^{\bar{\sigma},\bar{\sigma'}}_{\br, \br'}s_{\sigma'} = 0\;\;\;\forall \sigma, \sigma'.
\label{eq:TSGAcondition}
\end{equation}
\section{Tower of States from (Restricted) Spectrum Generating Algebras}\label{app:RSGAtower}
Here we show that the (Restricted) Spectrum Generating Algebras lead to the existence of a tower of exact eigenstates of the Hamiltonian.
We work with a Hamiltonian $H$ and ``root eigenstate" $\ket{\psi_0}$ from which the tower is generated by the application of $\ed$ operator, and use the definition of Eqs.~(\ref{eq:psindefn}) and (\ref{eq:Hndefn}).
We define a set of states $\{\ket{\psi_n}\}$ as
\begin{equation}
    \ket{\psi_n} \equiv (\ed)^n\ket{\psi_0},
\label{eq:psindefnapp}
\end{equation}
and a set of operators $\{H_n\}$ as
\begin{equation}
    H_0 \equiv H,\;\; H_{n+1} \equiv [H_n, \ed],\;\;\forall n \geq 0. 
\label{eq:Hndefnapp}
\end{equation} 
\begin{lemma}[SGA]\label{lem:SGA}
If the Hamiltonian $H$ and operator $\ed$ satisfy the conditions
\begin{enumerate}
    \item[(i)] $H\ket{\psi_0} = E_0\ket{\psi_0}$ 
    \item[(ii)] $[H, \ed] = \mE \ed$  
\end{enumerate}
then 
\begin{equation}
    H\ket{\psi_n} = (E_0 + n\mE)\ket{\psi_n}\;\;\textrm{or}\;\;\ket{\psi_n} = 0. 
\label{eq:SGAlemma}
\end{equation}
\end{lemma}
\begin{proof}
The proof proceeds straightforwardly via induction. 
Assuming $\ket{\psi_m}$ satisfies Eq.~(\ref{eq:SGAlemma}), we show $\ket{\psi_{m+1}}$ satisfies Eq.~(\ref{eq:SGAlemma}) provided it does not vanish. 
Using condition (ii), we obtain
\begin{eqnarray}
    &[H, \ed]\ket{\psi_m} = \mE \ket{\psi_m} \implies (H \ed - \ed H)\ket{\psi_m} = \mE\ket{\psi_m} \nn \\
    &\implies H \ed \ket{\psi_{m}} = \left(E_0 + m\mE + \mE\right)\ed \ket{\psi_{m}}.
\label{eq:SGAproof}
\end{eqnarray}
Thus, either $\ket{\psi_{m+1}} = 0$ or $H\ket{\psi_{m+1}} = \left(E_0 + (m+1)\mE\right)\ket{\psi_{m+1}}$.
Since Eq.~(\ref{eq:SGAlemma}) is satisfied for $m = 0$ (due to condition (i)), this concludes the proof.
\end{proof}
\begin{lemma}[RSGA-1]\label{lem:RSGA1}
If the Hamiltonian $H$, operator $\ed$, and state $\ket{\psi_0}$ such that $\ed\ket{\psi_0} \neq 0$ satisfy the conditions
\begin{enumerate}
    \item[(i)] $H\ket{\psi_0} = E_0\ket{\psi_0}$
    \item[(ii)] $[H, \ed]\ket{\psi_0} = \mE\ed\ket{\psi_0}$ $(i.e.\ H_1 \ket{\psi_0} = \mE\ket{\psi_1})$
    \item[(iii)] $[[H, \ed],\ed] = 0$ $(i.e.\ H_2 = 0)$
\end{enumerate}
then 
\begin{equation}
    H\ket{\psi_n} = (E_0 + n\mE)\ket{\psi_n}\;\;\textrm{or}\;\;\ket{\psi_n} = 0. 
\label{eq:RSGAlemma}
\end{equation}
\end{lemma}
\begin{proof}
The proof proceeds by induction on two levels. 
We first wish to show
\begin{equation}
    H_1 \ket{\psi_n} = \mE\ket{\psi_{n+1}} 
\label{eq:RSGAtoshow}
\end{equation}
For the purposes of induction, we assume Eq.~(\ref{eq:RSGAtoshow}) is valid for $\ket{\psi_m}$. 
Using condition (iii) we obtain
\begin{eqnarray}
    &H_2 \ket{\psi_m} = [H_1, \ed]\ket{\psi_m} = 0 \implies (H_1 \ed - \ed H_1)\ket{\psi_m} = 0 \nn \\
    &\implies H_1 \ket{\psi_{m+1}} = \mE\ket{\psi_{m+2}}.
\label{eq:RSGAproof}
\end{eqnarray}
Since Eq.~(\ref{eq:RSGAtoshow}) is satisfied for $m = 0$ (due to condition (ii)), this concludes the proof of Eq.~(\ref{eq:RSGAtoshow}).  
Using Eq.~(\ref{eq:RSGAtoshow}), we show Eq.~(\ref{eq:RSGAlemma}) by induction again.
Assuming Eq.~(\ref{eq:RSGAlemma}) holds for $\ket{\psi_m}$, using Eq.~(\ref{eq:SGAproof}) we can show $\ket{\psi_{m+1}}$ also satisfies it.
Since $\ket{\psi_0}$ satisfies Eq.~(\ref{eq:RSGAlemma}) (due to condition (i)), this concludes the proof.  
\end{proof}
\begin{lemma}[RSGA-M]\label{lem:RSGA2}
If the Hamiltonian $H$, operator $\ed$, and state $\ket{\psi_0}$ such that $(\ed)^{n}\ket{\psi_0} \neq 0$ for $n \leq M$ satisfy the conditions
\begin{enumerate}
    \item[(i)] $H\ket{\psi_0} = E_0\ket{\psi_0}$
    \item[(ii)] $H_1\ket{\psi_0} = \mE\ket{\psi_1}$ 
    \item[(iii)] $H_n\ket{\psi_0} = 0\;\;\forall n,\;\; 2 \leq n \leq M$
    \item[(iv)] $H_{M+1} = 0$ $(i.e.\ [H_{M}, \ed] = 0)$
\end{enumerate}
then 
\begin{equation}
    H\ket{\psi_n} = (E_0 + n\mE)\ket{\psi_n}\;\;\textrm{or}\;\;\ket{\psi_n} = 0. 
\label{eq:RSGA2lemma}
\end{equation}
\end{lemma}
\begin{proof}
We start with conditions (iii) and (iv) and note that they satisfy the conditions (i) and (ii) of Lemma~\ref{lem:SGA} with the replacements $H \rightarrow H_{M}$, $E_0 \rightarrow 0$, and $\mE \rightarrow 0$. 
Using Eq.~(\ref{eq:SGAlemma}), we arrive at
\begin{equation}
    H_{M} \ket{\psi_n} = 0 = [H_{M-1}, \ed]\ket{\psi_n}\;\;\forall\ n. 
\label{eq:HMcondition1}
\end{equation}
Further, using Eq.~(\ref{eq:SGAproof}) along with condition (i) of Lemma~\ref{lem:SGA} with the replacements $H \rightarrow H_{M-1}$, $E_0 \rightarrow 0$, $\mE \rightarrow 0$, as a consequence of Eq.~(\ref{eq:HMcondition1}) we obtain 
\begin{equation}
    H_{M-1}\ket{\psi_n} = 0 = [H_{M-2}, \ed]\ket{\psi_n}\;\;\forall\ n,
\label{eq:HMcondition2}
\end{equation}
which is the same as Eq.~(\ref{eq:HMcondition1}) with the replacements $H_M \rightarrow H_{M-1}$ and $H_{M-1} \rightarrow H_M$. 
Repeating the steps from Eq.~(\ref{eq:HMcondition1}) to Eq.~(\ref{eq:HMcondition2}) successively replacing $H_n \rightarrow H_{n-1}$ at each step, we finally arrive at 
\begin{equation}
    H_2 \ket{\psi_n} = 0 = [H_1, \ed]\ket{\psi_n}\;\;\forall\ n.
\label{eq:HMconditionfinal}
\end{equation}
The proof of Eq.~(\ref{eq:RSGA2lemma}) can then be completed following the same steps as the proof of Lemma~\ref{lem:RSGA1}.
\end{proof}
\section{Quantum Scars with Translation Invariance}\label{app:scarsTI}
Here we show that the states of the tower $\{(\ed)^n\ket{\Omega}\}$ lie in the bulk of the spectrum of the Hamiltonian $(\hhub + \hI_{M+1})$ even if translation, inversion, spin-flip, and $S_z$ symmetries are preserved.
We start with the one-dimensional Hubbard model of even length $L$ with periodic boundary conditions, and add to it a small perturbation $\hI_{M+1}$ that breaks the integrability of spin $SU(2)$ symmetry and $\hhub$, but preserves the spin-flip and translation symmetries.
We have numerically verified that such an $\hI_{M+1}$ can be found for generic choices of parameters $\{V^{\{\sigma_j\}}_{\{\br_j\}}\}$ that satisfy Eq.~(\ref{eq:Imcondition}). 
The states $(\ed)^n\ket{\Omega}$ for even $n$ has the following quantum numbers: momentum $k = 0$, charge $Q = 2n$, spin $S_z = 0$, inversion $I = +1$ (for inversion about a site), spin-flip $P_z = +1$.
To determine whether these $\eta$-pairing states lie in the bulk of the energy spectrum within their quantum number sectors, we can consider the non-interacting ferromagnetic eigenstates since their energies can be analytically obtained. 
Non-interacting ferromagnetic eigenstates with charge $Q = 2n$ and $k = 0$ can be constructed by occupying the single-particle spectrum of the quadratic part of the Hamiltonian $\hhub$ with $2n$ $\uparrow$ spins such that the total momentum adds up to $0$.
Such states with inversion quantum number $I = +1$ can be constructed by having an even number of pairs of occupied single-particle levels with momenta $k$ and $-k$, which can be realized when $n$ is even.
As a consequence of the spin $SU(2)$ symmetry of $\hhub$, states with $S_z = 0$ can be obtained by applying the spin lowering operator on these ferromagnetic states. The $S_z = 0$ states thus obtained are also guaranteed to have spin-flip quantum number $P_z = +1$ since they are part of the ferromagnetic multiplet, and the inversion quantum number remains unchanged by spin lowering.
The lowest energy ferromagnetic state with these quantum number constraints is constructed by occupying the lowest $2n$ single-particle eigenstates except the $k = 0$ level by $\uparrow$ spins.
Similarly, the highest energy ferromagnetic state with these quantum number constraints is built by occupying the highest $2n$ single-particle eigenstates except the $k = \pi$ level by $\uparrow$ spins.
Their energies thus read
\begin{eqnarray}
    E_{\pm} &=& -2n\mu \mp 2t\ \left[\left(\sumal{j = -n}{n}{\cos\left(\frac{2\pi j}{L}\right)}\right) - 1\right] \nn \\
    &=& -2n \mu \mp 2t\ \left(\csc\left(\frac{2\pi}{L}\right)\sin\left(\frac{(2n+1)\pi}{L}\right) - 1\right).\nn \\
\label{eq:EpmTI}
\end{eqnarray}
Similar to the case with disorder discussed in Sec.~\ref{sec:quantumscarconnection}, for small perturbations $\hI_{M+1}$, we can use these non-interacting states to estimate the energies of the lowest and highest excited state restricted to a given quantum number sector. 
The state $(\ed)^n\ket{\Omega}$, the eigenstate with $n$ doublons with energy $n (U - 2\mu)$, certainly lies in the bulk of the spectrum of its own quantum number sector if $E_- < n(U - 2\mu) < E_+$, or,
\begin{eqnarray}
    &-\frac{2t}{n}\ \left(\csc\left(\frac{2\pi}{L}\right)\sin\left(\frac{(2n+1)\pi}{L}\right) - 1\right) < U \nn \\
    &< \frac{2t}{n}\ \left(\csc\left(\frac{2\pi}{L}\right)\sin\left(\frac{(2n+1)\pi}{L}\right) - 1\right).
\label{eq:TIbound}
\end{eqnarray}
For a finite density of doublons in the thermodynamic limit ($n/L = \rho$, while $n, L \rightarrow \infty$), we recover the bound of Eq.~(\ref{eq:bulkconditionrho}).
\section{Quantum Scars in $\boldsymbol{D}$ dimensions}\label{app:scarshigherD}
In this appendix, we obtain the conditions for which the states of the tower $\{(\ed)^n\ket{\Omega}\}$ are in the bulk of the spectrum of the Hamiltonian $(\hhub + \hI_{M+1})$ in $D$-dimensions. 
Consider a system of $L \times L \times \cdots \times L$ sites in $D$ dimensions, and the state $(\ed)^n\ket{\Omega}$ consisting of $n$ doublons.
To obtain the conditions for $(\ed)^n\ket{\Omega}$ to lie in the bulk of the spectrum for a small perturbation $\hI_{M + 1}$, it is sufficient to obtain the energies $E_-$ and $E_+$ of the lowest and highest ferromagnetic non-interacting eigenstates with charge $Q = 2n$, as discussed in Sec.~\ref{sec:quantumscarconnection}. 
To do so, we directly work in the continuum limit in momentum space and with a finite density of doublons, i.e.
\begin{equation}
    \rho \equiv \frac{n}{L^D}.
\label{eq:doublondensity}
\end{equation}
The single-particle density of states $f(\bk)$ reads
\begin{equation}
    f(\bk) = \left(\frac{L}{2\pi}\right)^D,\;\;\; \int{\mathrm{d}^D\bk\ f(\bk)} = L^D.
\label{eq:DOS}
\end{equation}
Assuming a spherical Fermi surface, the Fermi momentum $k_F$ by filling the lowest $2n$ single-particle levels satisfies the relation
\begin{equation}
    \int_{|\bk| < k_F}{\mathrm{d}^D\bk\ f(\bk)} = \left(\frac{L}{2\pi}\right)^D \frac{\pi^{\frac{D}{2}} k_F^D}{\Gamma\left(\frac{D}{2} + 1\right)} = 2n,
\label{eq:fermicondition}
\end{equation}
where we have used the expression for the volume of $D$-dimensional sphere. 
We thus obtain 
\begin{eqnarray}
    k_F &=& \frac{2\sqrt{\pi}}{L}\left(2 n\ \Gamma\left(\frac{D}{2} + 1\right)\right)^{\frac{1}{D}} \nn \\
    &=& 2\sqrt{\pi}\left(2 \rho\ \Gamma\left(\frac{D}{2} + 1\right)\right)^{\frac{1}{D}}.
\label{eq:kfermi}
\end{eqnarray}
Note that the spherical approximation of the Fermi surface in Eq.~(\ref{eq:fermicondition}) breaks down for sufficiently large $n$ when the Fermi surface is close to the edges of the Brillouin zone, i.e. when the Fermi momentum $k_F$ obtained using Eq.~(\ref{eq:kfermi}) is comparable to $\pi$. 
Thus, the calculations in this section are strictly valid only when $k_F \ll \pi$, or when the doublon density $\rho$ satisfies
\begin{equation}
    \rho \ll \frac{\pi^{\frac{D}{2}}}{2^{D+1}\Gamma\left(\frac{D}{2} + 1\right)}.
\label{eq:doubdensity}
\end{equation} 
However, we expect that similar arguments work for the larger densities as well. 
Since the dispersion relation of the quadratic part of the $D$-dimensional Hubbard models is given by Eq.~(\ref{eq:dispersion}), the energy of the states obtained by filling all single-particle momentum levels with $|\bk| < k_F$ is given by
\begin{eqnarray}
    E_- &=& -2n\mu - 2 t D\int_{|\bk| < k_F}{\mathrm{d}^D\bk\ f(\bk)\ \cos\left(k_i\right)} \nn \\
    &=& -2n\mu - 2 t D \left(\frac{L}{2\pi}\right)^D \underbrace{\int_{|\bk| < k_F}{\mathrm{d}^D\bk\ \cos\left(k_i\right)}}_{\equiv \mathcal{I}_D(k_F)}, \nn \\
\label{eq:Energylow}
\end{eqnarray}
where $k_i$ is the component of $\bk$ along any axis. 
Evaluating the integral $\mathcal{I}_D(k_F)$ in $D$-dimensional spherical coordinates, we obtain
\begin{equation}
    \mathcal{I}_D(k_F) = (2\pi k_F)^{\frac{D}{2}} \mathcal{J}_{\frac{D}{2}}(k_F),
\label{eq:integral}
\end{equation}
where $\mathcal{J}_\alpha(x)$ is the $\alpha$-th order Bessel function of the first kind.
Note that $\mathcal{J}_\alpha(x)$ for $\alpha \in \mathbb{Z} + \frac{1}{2}$ can be expressed in terms of trigonometric functions. 
Thus, for $D = 1$ for example, we obtain
\begin{equation}
    \mathcal{I}_1(k_F) = 2 \sin(k_F).
\end{equation}
For the highest energy state, similar to Eq.~(\ref{eq:Energylow}), we obtain
\begin{equation}
    E_+ =  -2n\mu + 2 t D \left(\frac{L}{2\pi}\right)^D \mathcal{I}_D(k_F).
\label{eq:Energyhight}
\end{equation}
The state $(\ed)^n\ket{\Omega}$ with $n$ doublons is thus guaranteed to be in the bulk of the spectrum if $E_- < n(U - 2\mu) < E_+$, or, 
\begin{equation}
    - \frac{2 D}{(2\pi)^D}\frac{\mathcal{I}_D(k_F)}{\rho}  < \frac{U}{t} < \frac{2 D}{(2\pi)^D}\frac{\mathcal{I}_D(k_F)}{\rho}, 
\label{eq:genUtbound}
\end{equation}
where $\mathcal{I}_D(k_F)$ and $k_F$ are defined in Eqs.~(\ref{eq:integral}) and (\ref{eq:kfermi}) respectively. Note that we recover Eq.~(\ref{eq:bulkconditionrho}) by setting $D = 1$ in Eq.~(\ref{eq:genUtbound}).
For $D = 2$, the bound reads
\begin{equation}
    -\sqrt{\frac{8}{\pi\rho}} \mathcal{J}_1(\sqrt{8\pi\rho}) < \frac{U}{t} < \sqrt{\frac{8}{\pi\rho}} \mathcal{J}_1(\sqrt{8\pi\rho}).
\label{eq:2Dbound}
\end{equation}
\section{RSGAs in the AKLT Family of Quantum Scarred Hamiltonians}\label{app:RSGAscars}
In this section, we show that the AKLT family of quantum scarred Hamiltonians studied in Refs.~\cite{mark2020unified, moudgalya2020large} admit RSGAs of order $M = 2$. 
Throughout this section, we use the notation $\ket{J_{j,m}}$ to denote a total angular momentum eigenstates of two spin-1's with total angular momentum quantum number $j$, $0 \leq j \leq 2$, and its $z$-projection quantum number $m$, $-j \leq m \leq j$.
We refer readers to Ref.~\cite{moudgalya2020large} for details of the notation.
The one-dimensional family of quantum scarred spin-1 Hamiltonians (including the spin-1 AKLT chain) on a system size of $L$ derived in Refs.~\cite{mark2020unified, moudgalya2020large} is given by
\begin{eqnarray}
    &\haklt = \sumal{j = 1}{L}{\hh_{j,j+1}},\nn \\
    &\hh_{j, j+1} = \mE\left(\ket{J_{2,1}}\bra{J_{2,1}} + \ket{J_{2,2}}\bra{J_{2,2}}\right) \nn \\
    &+ \sumal{m, n = -2}{0}{z^{(m,n)}_j (\ket{J_{2,m}}\bra{J_{2,n}})}.  
\label{eq:AKLTfamilHamil}
\end{eqnarray}
where $(z^{(m,n)}_j) = (z^{(m,n)}_j)^\ast$.
As discussed in Refs.~\cite{mark2020unified} and \cite{moudgalya2020large}, for an even system size $L$ and periodic boundary conditions, the Hamiltonian of Eq.~(\ref{eq:AKLTfamilHamil}) contains a tower of quantum scars from a root eigenstate $\ket{G}$,
\begin{equation}
    \haklt (\paklt)^n\ket{G} = 2n\mE(\paklt)^n\ket{G},\;\; 0\leq n \leq \frac{L}{2},
\label{eq:RSGA0}
\end{equation}
where $\ket{G}$ is the spin-1 AKLT ground state~\cite{aklt1987, Moudgalya2018a}, and 
\begin{equation}
    \paklt = \sumal{j = 1}{L}{(-1)^j (S^+_j)^2}, 
\label{eq:mPdefn}
\end{equation}
which forms the analogue of the $\ed$ operator in the Hubbard models discussed in the main text. 
The spin-1 AKLT Hamiltonian~\cite{aklt1987, Moudgalya2018a} is recovered from Eq.~(\ref{eq:AKLTfamilHamil}) by setting~\cite{moudgalya2020large}
\begin{equation}
    \mE = 1,\;\;z^{(m,n)}_j = \delta_{m,n}.
\label{eq:AKLTparams}
\end{equation}
We first compute the commutator
\begin{eqnarray}
    &[\haklt, \paklt] = \sumal{j = 1}{L}{(-1)^j [\hh_{j,j+1}, (S^+_j)^2 - (S^+_{j+1})]} \nn \\
    &\equiv 2\mE \paklt + \sumal{j =1 }{L}{(-1)^j \hh^{(1)}_{j,j+1}},\nn \\
    &\hh^{(1)}_{j,j+1} = -2\sumal{n = -2}{0}{\left((z^{(-1,n)}_j - \mE)\ket{J_{1,1}}\bra{J_{2,n}}\right.} \nn \\
    &{\left.+ \sqrt{2} (z^{(-2,n)}_j - \mE)\ket{J_{1,0}}\bra{J_{2,n}}\right)},\nn\\
\label{eq:Hcomm}
\end{eqnarray}
where we have used Eq.~(\ref{eq:AKLTfamilHamil}), and~\cite{mark2020unified}
\begin{eqnarray}
    &(S^+_j)^2 - (S^+_{j+1})^2 = -2\left(\ket{J_{2,1}}\bra{J_{1,-1}} +\sqrt{2}\ket{J_{2,2}}\bra{J_{1,0}} \right.\nn \\
    &\left.+ \ket{J_{1,1}}\bra{J_{2,-1}} + \sqrt{2}\ket{J_{1,0}}\bra{J_{2,-2}}\right).
\label{eq:SpSpdiff}
\end{eqnarray}
Using Eqs.~(\ref{eq:AKLTfamilHamil}) and (\ref{eq:Hcomm}), it is apparent that $\haklt$ can be decomposed as
\begin{equation}
    \haklt = \haklt_{\textrm{SGA}} + \vaklt, 
\end{equation}
where
\begin{eqnarray}
	&\haklt_{\textrm{SGA}} \equiv \sumal{j = 1}{L}{\mE\left(\ket{J_{2,1}}\bra{J_{2,1}} + \ket{J_{2,2}}\bra{J_{2,2}}\right.} \nn \\
	&\left.- \ket{J_{2,-1}}\bra{J_{2,-1}} - \ket{J_{2,-2}}\bra{J_{2,-2}}\right),
\label{eq:AKLTdecomp}
\end{eqnarray}
and it admits an exact SGA, i.e.
\begin{equation}
	[\haklt_{\textrm{SGA}}, \paklt] = 2\mE\paklt. 
\label{eq:AKLTSGA}
\end{equation}
Further, we note that
\begin{equation}
    \hh^{(1)}_{j,j+1}\ket{G} = 0,
\label{eq:h1vanish}
\end{equation}
since the AKLT ground state does not have a total spin 2 component over neighboring sites~\cite{aklt1987}.
Using Eq.~(\ref{eq:Hcomm}), we thus obtain
\begin{equation}
    [\haklt, \paklt]\ket{G} = 2\mE\paklt\ket{G}.
\label{eq:RSGA1}
\end{equation}
We further compute the next commutator
\begin{eqnarray}
    &[[\haklt, \paklt], \paklt] = \sumal{j = 1}{L}{[h^{(1)}_{j,j+1}, (S^+_j)^2 - (S^+_{j+1})^2]} \equiv \sumal{j = 1}{L}{\hh^{(2)}_{j,j+1}} \nn \\
    &\hh^{(2)}_{j,j+1} = -4\sqrt{2}\sumal{n = -2}{0}{(z^{(-2,n)}_j-\mE) \ket{J_{2,2}}\bra{J_{2,n}}}.
\label{eq:H2comm}
\end{eqnarray}
Similar to Eq.~(\ref{eq:h1vanish}), we obtain
\begin{equation}
    \hh^{(2)}_{j,j+1}\ket{G} = 0,\;\;\textrm{and}\;\;[[\haklt, \paklt], \paklt]\ket{G} = 0. 
\label{eq:RSGA2}
\end{equation}
Using Eqs.~(\ref{eq:H2comm}) and (\ref{eq:SpSpdiff}), we further obtain
\begin{eqnarray}
    [[[\haklt, \paklt], \paklt],\paklt] &=& \sumal{j = 1}{L}{(-1)^j[\hh^{(2)}_{j,j+1}, (S^+_j)^2 - (S^+_{j+1})^2]} \nn \\
    &=& 0. 
\label{eq:RSGA3}
\end{eqnarray}
Using Eqs.~(\ref{eq:RSGA0}), (\ref{eq:RSGA1}), (\ref{eq:RSGA2}), and (\ref{eq:RSGA3}), we obtain that the family of Hamiltonians of Eq.~(\ref{eq:AKLTfamilHamil}) admit an RSGA of order $M = 2$ (see Lemma~\ref{lem:RSGA2}) with $\ket{\psi_0} = \ket{G}$, $E_0 = 0$, $\mE = 2\mE$, and $\ed = \paklt$.
Similarly, we can verify that the same algebraic structure holds for the single-site quasiparticle family of scarred Hamiltonians studied in Ref.~\cite{moudgalya2020large}, the one-dimensional spin-$S$ AKLT Hamiltonians~\cite{Moudgalya2018a} and the associated family of scarred Hamiltonians discussed in Ref.~\cite{mark2020unified}.
\bibliography{eta_pairing}
\end{document}